\documentclass[11pt]{article}
\usepackage{amsfonts,amsmath,amsthm,fullpage,color}
\newtheorem{theorem}{Theorem}[section]

\newtheorem{lemma}[theorem]{Lemma}
\theoremstyle{definition}
\newtheorem{example}{Example}

\newcommand{\Pro}{\textsf{Prover}}
\newcommand{\POR}{\textsf{POR}}
\newcommand{\Adv}{\textsf{Adversary}}
\newcommand{\Ver}{\textsf{Verifier}}
\newcommand{\Ext}{\textsf{Extractor}}
\newcommand{\eff}{\mathbb{F}}
\newcommand{\Pos}{\ensuremath{\mathsf{Possible}}}
\newcommand{\dist}{\ensuremath{\mathsf{dist}}}

\newcommand{\succp}{{\sf succ}({\cal P})}

\begin{document}

\title{A coding theory {foundation} for the analysis of {general}  unconditionally secure 
proof-of-retrievability schemes for cloud storage}
\author{Maura~B.~Paterson\\
Department of Economics,
Mathematics and Statistics\\ Birkbeck, University of London,
Malet Street, London WC1E 7HX, UK
\and
Douglas~R.~Stinson\thanks{D.~Stinson's research is supported by NSERC discovery grant 203114-11}
\\David R. Cheriton School of Computer Science\\ University of Waterloo,
Waterloo, Ontario, N2L 3G1, Canada
\and
Jalaj Upadhyay
\\David R. Cheriton School of Computer Science\\ University of Waterloo,
Waterloo, Ontario, N2L 3G1, Canada}
\date{\today}

\maketitle

\begin{abstract}
There has been considerable recent interest in ``cloud storage'' wherein
a user asks a server to store a large file. One issue is whether the user can
verify that the server is actually storing the file, and typically a challenge-response
protocol is employed to convince the user that the file is indeed being stored
correctly. The security of these schemes is phrased in terms of an extractor which will
recover or retrieve the file given any ``proving algorithm'' that has a
sufficiently high success probability.

This paper treats proof-of-retrievability schemes in the model of unconditional
security, where an adversary has unlimited computational power. In this case retrievability
of the file can be modelled as error-correction in a certain code. We provide a
general analytical framework for such schemes that yields exact (non-asymptotic) reductions
that precisely quantify conditions for extraction to succeed as a function of the success
probability of a proving algorithm, and we apply this analysis to several
archetypal schemes.  In addition, we provide a new methodology for the analysis
of keyed \POR \  schemes in an unconditionally secure setting, and use it to
prove the security of a modified version of a scheme due to Shacham and Waters under a slightly
restricted attack model, thus providing the first example of a keyed \POR \ scheme
with unconditional security. {We also show how classical statistical techniques can be used to evaluate whether the responses of the prover are 
accurate enough to permit successful extraction.}
{Finally, we prove a new lower bound on storage and communication complexity of \POR \ schemes.}
\end{abstract}

\section{Introduction to Proof-of-retrievability Schemes}

In a {\it proof-of-retrievability scheme} (\POR \ scheme) \cite{BJO,DVW,JK,SW}, a user asks a
server to store a (possibly large) file.
The file is divided into {\it message blocks} that we view as elements
of a finite field.
Typically, the file will be ``encoded'' using a code such as
a Reed-Solomon code. The code provides redundancy, enabling erasures
or corrupted message blocks to be corrected.

In order for the user to be ensured that the file is being
stored correctly on the server, a challenge-response protocol is periodically invoked by the
user, wherein the server (the \Pro ) must give a correct response
to a random challenge chosen by the user
(the \Ver).
This response will typically be a function of one or more message blocks.
We do not assume that the user is storing the file. Therefore, in the basic version of
a \POR \ scheme,
the user must precompute and store a sufficient number of challenge-response pairs,
before transmitting the file to the server. After this is done, the user erases the file
but retains the precomputed challenge-response pairs. Such schemes are termed {\it bounded-use} schemes
in \cite{DVW}.
{An alternative is to use a {\it keyed} (or {\it unbounded-use}) scheme, which permits an arbitrary 
number
of challenges to be verified (i.e., the number of challenges does not need to be pre-determined by the 
user). We will discuss these a bit later.
Finally, the user could retain a copy of
the file if desired,} in which case the responses do not need to be precomputed.
This might be done if the server is just
being used to store a ``back-up'' copy of the file.

We wish to quantify the security afforded the user by engaging in the challenge-response
protocol (here, ``security'' means that the user's file can be correctly retrieved by
the user). The goal is that a server
who can respond correctly to a large proportion
of challenges somehow ``knows'' (or can compute)
the contents
of the file (i.e., all the message blocks).
This is formalised through the notion of
an {\it extractor}, which takes as input a description of a ``proving algorithm'' $\mathcal{P}$ for
a certain unspecified file, and then outputs the file. Actually, for the schemes
we study in this paper, the extractor only needs {\it black-box}
access to the proving algorithm. The proving algorithm is created
by the server, who is regarded as the adversary in this game.
It should be the case that
a proving algorithm that is correct with a probability  that is sufficiently close to $1$
will allow the extractor to determine the correct file.
The probability that the proving algorithm $\mathcal{P}$ gives a correct response
for a randomly chosen challenge is denoted by $\mathsf{succ}(\mathcal{P})$. We assume
that $\mathcal{P}$ always gives some response, so it follows that
$\mathcal{P}$ will give an incorrect response with probability $1 - \mathsf{succ}(\mathcal{P})$.

{To summarise, we list} the components in a \POR \ scheme and the 
extractor-based security definition we use.
{Note that we are employing standard models developed in the literature; for a more detailed discussion,
see \cite{BJO,JK,SW}.}
\begin{itemize}
\item The \Ver \ has a  message ${m \in (\eff_q)^k}$ which he redundantly encodes as ${M \in (\eff_q)^n}$.
\item $M$ is given to the \Pro. In the case of a keyed scheme,
the \Pro \ may also be supplied with an additional tag, $S$.
\item The \Ver \ retains appropriate information to allow him to verify responses.
This may or may not include a key $K$.
\item Some number of challenges and responses are carried out by the \Pro \ and \Ver.
In each round, the \Ver \ chooses a challenge $c$ and gives it to the \Pro, and the \Pro \
computes a response $r$ which is returned to the \Ver. The \Ver \ then verifies if the response
is correct.
\item The computations of the \Pro \ are described in a proving algorithm $\mathcal{P}$.
\item The success probability of $\mathcal{P}$ is the probability that it gives a correct
response when the challenge is chosen randomly.
\item The \Ext \ is given $\mathcal{P}$ and (in the case of a keyed scheme) $K$,
and outputs an unencoded message $\hat{m}$. Extraction succeeds if $\widehat{m} = m$.
\item The security of the \POR \ scheme is quantified by proving a statement of the form
``the \Ext \ succeeds with probability at least $\delta$ whenever the success probability of
$\mathcal{P}$ is at least $\epsilon$''. In this paper, we only consider schemes where $\delta = 1$,
that is, where extraction is always successful.
\end{itemize}

\subsection{{Previous Related Work}}


{Blum et al~\cite{BEGKN94} introduced the concept of memory checking. 
They formalized a model where performing any 
sequence of store and request operations on a  remote server behaves similarly to local storage.
Two seminal papers important to the development of \POR \ schemes are
\cite{LEBB,NR}.  Lillibridge {\it et al}.\ \cite{LEBB} first considered the problem of 
creating a backup of a large file by redundantly encoding the using
an erasure code and distributing pieces of the file to one or more servers.
Naor and Rothblum \cite{NR} studied {\it memory checkers} which use
message authentication techniques to verify
if a file is stored correctly on a remote server. They also consider
{\it authenticators}, which allow a verifier to interact with the server and
reconstruct the file provided that a sufficient number of ``audits''
are all correct.}

{The concept of {\it proof-of-retrievability} is due to Juels and Kaliski
\cite{JK}. A  \POR \ scheme incorporates a challenge-response protocol
in which a verifier can check that a file is being stored correctly, 
along with an {\it extractor} that will actually reconstruct the file,
given the algorithm of a ``prover'' who is able to correctly respond
to a sufficiently high percentage of challenges.}

{
There are also papers that describe the closely related (but slightly weaker) idea of a
{\it proof-of-data-possession scheme} (\textsf{PDP} scheme), e.g., \cite{Aten}.
A \textsf{PDP} scheme permits the possibility that 
not all of the message blocks can be reconstructed.}
{Atieniese {\it et al.}\ \cite{Aten} also introduced the idea 
of using {\it homomorphic authenticators} to reduce the communication complexity of the system. 
Shacham and Waters~\cite{SW} showed that the scheme of 
Ateniese {\it et al.}\ can be transformed in to 
a \POR \ scheme by constructing an extractor that extracts the file 
from the responses of the server on audits.
 
Bowers, Juels, and Oprea~\cite{BJO} used error-correcting codes, in particular the idea 
of an ``outer'' and an ``inner'' code (in much the same 
vein as concatenated codes), to get a good balance between 
the server storage overhead and computational overhead in 
responding to the audits. {A paper that discusses a 
coding-theoretic approach in the setting of storage enforcement is
\cite{HKRU}; this paper makes use of list-decoding techniques, 
but it concentrates
on the storage requirements of the server.}

Dodis, Vadhan and Wichs \cite{DVW}
provide the first example of an unconditionally secure {\POR } \ scheme,
also constructed from an error-correcting code, with extraction performed through
list decoding in conjunction with the use of an
almost-universal hash function. 
We discuss this particular paper further in Section \ref{comparison.sec}.}

\subsection{{Our Contributions}}

In this paper, we
treat the general construction of extractors for \POR \ schemes, with the aim of
establishing the precise conditions under which extraction is possible in the
setting of {\it unconditional security}, where the adversary is assumed to have
unlimited computational capabilities.  In this setting, it turns out that extraction
can be interpreted naturally as nearest-neighbour decoding in a certain code
(which we term a ``response code''). {Previously, error-correcting codes have
been used in specific constructions of \POR \ schemes; here, we propose that 
error-correcting codes constitute the natural foundation to construct as well as analyse arbitrary \POR \
schemes.}

{There are several advantages of studying unconditionally secure \POR \ schemes.
First, the schemes are easier to understand and analyse because we are not making use of 
any additional cryptographic primitives or unproven assumptions (e.g., PRFs, signatures, bilinear pairings, 
MACS, hitting samplers, random oracle model, etc.).
This allows us to give very simple exact analyses of various schemes.
Secondly, the essential role of error-correcting codes in the design and analysis
of \POR \ schemes becomes clear: codes are not just a method of 
constructing \POR \ schemes; rather, every \POR \ scheme gives rise to a code in a natural way.}

The success of the extraction process usually depends
on the distance of the code used to initially encode the file; when the distance of this code
is increased, the extraction process will be successful for less successful provers,
thus increasing the security afforded the user. {As we mentioned earlier,} 
we quantify the security of a \POR \ scheme  by specifying
a value $\epsilon$ and proving that the extraction process will
{\it always} be able to extract the file, given a prover $\mathcal{P}$ with success probability
$\mathsf{succ}(\mathcal{P}) > \epsilon$.  (In some other papers, weaker types of extractors are studied, such
as extractors that succeed with some specified probability less than 1, or extractors that
only recover some specified fraction of the original data.)  This allows us to derive conditions which
guarantee that extraction will succeed, and to compute exact (or tight) bounds,
as opposed to the mainly asymptotic bounds appearing in \cite{DVW}.

We exemplify our approach by considering several archetypal \POR \  schemes. We consider 
{keyed schemes as well as keyless schemes.}  {In Section \ref{keyless.sec}, we first} 
consider the fundamental case where the
server is just required to return one or more requested message blocks.
Then we progress to a more general scheme where the server must compute
a specified linear combination of certain message blocks. Both of these are
keyless schemes. 

{In Section \ref{keyed.sec}, we} investigate 
the Shacham-Waters scheme \cite{SW}, which is a  keyed scheme,
modified appropriately to fit the setting of unconditional security.
For this scheme, we note that unconditional security can be achieved {only if the prover 
does not have access to a verification oracle.} It is also necessary to {analyse} the success 
probability of
a proving algorithm {in the average case,  over the set of keys that are consistent with
the information given to the prover}.  This new analytical approach is the first
to {allow} the {construction} of a {secure} keyed \POR \ scheme in the unconditionally secure setting.

{In Section \ref{numer.sec}, we look more closely at the numerical conditions we have derived 
for the various schemes we studied and we provide some useful comparisons
and estimates.}

{We desire that} successful extraction can be accomplished whenever
 $\mathsf{succ}(\mathcal{P})$ exceeds some prespecified threshold.
But this raises the question as to how the user is able to determine
(or estimate) $\mathsf{succ}(\mathcal{P})$.
In many practical \POR \  schemes, the only interaction a user has with the server is through the challenge-response
protocol.  We show {in Section \ref{confidence.sec}} that classical statistical techniques can be used to provide
a systematic basis for evaluating {whether the responses of the prover are 
accurate enough to permit successful extraction}.

{The main overhead of keyed (unbounded-use) unconditionally secure schemes
(relative to computationally secure schemes) is in the storage requirements of the user.  While this
may be problematic in some circumstances, there are significant applications, such as remote backup
services, for which such schemes do genuinely represent a practical solution.
We also show in Section \ref{lower.sec}
that a significant additional storage requirement cannot be avoided
in this setting, by proving a new information-theoretic
lower bound on storage and communication requirements of 
\POR \ schemes. Our new bound is an improvement of the information-theoretic
lower bound for memory checkers and authenticators proven in \cite{NR}.}

We note that our goal is not to identify a single ``best'' \POR \  scheme,
but rather to provide a useful new methodology to analyse the exact security of various \POR \  schemes
in the unconditionally secure setting.   The \POR\ problem is one that lends
itself naturally to analysis in the unconditionally secure model; the schemes we consider are natural
examples of \POR\ protocols that do not require cryptographic building blocks
depending on computational assumptions.  

For reference, we
provide a list of notation used in this paper in Table \ref{notation.tab} in the Appendix.

\subsection{{Comparison with Dodis, Vadhan and Wichs \cite{DVW}}}
\label{comparison.sec}

{Our work is most closely related to that of Dodis, Vadhan and Wichs \cite{DVW}.
In \cite{DVW}, it is stated that ``there is a clear relation between our problem and
the erasure/error decoding of error-correcting codes''. Our paper is in some sense
a general exploration of these relations, whereas \cite{DVW} is mainly devoted
to a specific construction for a \POR \ scheme that is based on Reed-Solomon codes.}

{We can highlight several differences between our approach and that of \cite{DVW}:
\begin{itemize}
\item {In the setting of unconditional security, the paper \cite{DVW} only provides
bounded-use schemes. Our scheme is the first unbounded-use scheme in this setting.}
\item The paper \cite{DVW} mainly uses a (Reed-Solomon) code to construct a specific  \POR \ scheme.
In contrast, we are studying the connections between an arbitrary \POR \ scheme 
and the distance of the (related) code that describes the behaviour of the scheme on the possible queries
that can be made to the scheme. Stated another way, our approach is to 
derive a code from a \POR \ scheme, and then to prove security properties of the \POR \ scheme 
as a consequence of properties of this code.
\item The paper \cite{DVW} uses various tools and algorithms to construct their \POR \ schemes, including
Reed-Solomon codes, list decoding, {almost-universal} hash families, 
and hitting samplers based on expander graphs. 
We just use an error-correcting code in our analyses.
\item We base our analyses on nearest-neighbour decoding (rather than list decoding,
which was used in \cite{DVW}), and we present conditions under
which extraction will succeed with probability equal to 1 (in \cite{DVW}, extraction succeeds with
probability close to 1, depending in part on properties of a certain class of hash functions used
in the protocol).
\item {The ``POR codes'' in \cite{DVW} are actually protocols that
consist of challenges and responses 
involving a prespecified number of message blocks; we allow challenges in which the 
responses depend on an arbitrary number of message blocks.}
\item All our analyses are exact and concrete, whereas the analyses in \cite{DVW} are asymptotic.
\end{itemize}
}

\section{Analysis of Several Keyless  Schemes}
\label{keyless.sec}

\subsection{A Basic   Scheme}

As a ``warm-up'', we illustrate our coding theory approach by analysing a simple \POR \  scheme, which
we call the \textsf{Basic  Scheme}. From now on, we usually refer to the
user as the \Ver \ and the server as the \Pro. The basic scheme is  presented in Figure \ref{fig-basic}.

\begin{figure}[tb]
\caption{\textsf{Basic  Scheme}}\label{fig-basic}
\begin{center}
\fbox{
    \begin{minipage}{15cm}
\begin{description}
\item [initialisation] \mbox{\quad} \\ Given a \emph{message} $m \in \mathcal{M}$, encode $M$ as
$e(m) = M \in \mathcal{M}^*$.  The \emph{message encoding function}
$e: \mathcal{M} \rightarrow \mathcal{M}^*$ is a public bijection.
We will assume that $\mathcal{M} = (\eff_q)^k$ and
$\mathcal{M}^* \subseteq (\eff_q)^n$, where $q$ is a prime power and $n \geq k$.
We write $M = (m_1, \dots , m_n)$, where the components $m_1, \dots , m_n$ are
termed \emph{message blocks}.

The \Ver \ gives the \emph{encoded message} $M$ to the
\Pro. The \Ver \ also generates a random \emph{challenge} $c \in \{1, \dots , n\}$ and
stores $c$ and the message block $m_c$.

\item [challenge-response] \mbox{\quad} \\
The \Ver \ gives the challenge $c$ to the \Pro.
The \Pro \ responds with the message block $r= m_c$. The \Ver \ checks that the
\emph{response} $r$ returned by the \Pro \ matches the stored value $m_c$.
\end{description}
\end{minipage}
  }
\end{center}
\end{figure}



{Here is the} adversarial model we use to analyse the security of  a keyless \POR \  scheme:  
{\begin{enumerate}
\item The \Adv \  is given the set of encoded messages $\mathcal{M}^*$.
\item The \Adv \   selects $m \in \mathcal{M}$.
\item The \Adv \   outputs a deterministic\footnote%
{{We note that, without loss of generality, the proving algorithm
can be assumed to be deterministic. This follows from the observation 
that any probabilistic proving algorithm can be replaced
by a deterministic algorithm relative to which the success of the extractor 
defined in Figure \ref{fig1} will not be increased.}} 
\emph{proving algorithm} for $m$, denoted $\mathcal{P}$.
\end{enumerate} }
%

The {\it success probability} of $\mathcal{P}$ is defined to be
\[\mathsf{succ}(\mathcal{P}) = \mathbf{Pr} [\mathcal{P}(c) = m_c],\]
where $e(m) = (m_1, \dots , m_n)$ and this probability is computed over
a challenge $c \in \{1, \dots , n\}$ chosen uniformly at random.

Now we wish to construct an \Ext \ that will take as input a proving algorithm
$\mathcal{P}$ for some unknown message $m$. The \Ext \ will output a message
$\widehat{m} \in \mathcal{M}$. We say that the \POR \ scheme is {\it secure} provided that
$\widehat{m} = m$ whenever  the success probability of $\mathcal{P}$ is sufficiently close to $1$.

The \Ext \ for the \textsf{Basic  Scheme} is presented in Figure \ref{fig1}.
\begin{figure}[tb]
\caption{\Ext \ for the \textsf{Basic    Scheme}}\label{fig1}
\begin{center}
\fbox{
    \begin{minipage}{15cm}
    \begin{enumerate}
\item On input $\mathcal{P}$, compute the vector $M' = (m'_1, \dots , m'_n)$,
where $m'_c = \mathcal{P}(c)$ for all $c \in \{1, \dots , n\}$ (i.e., $m'_c$ is
the response computed by $\mathcal{P}$ when it is given
the challenge $c$).
\item Find $\widehat{M} \in \mathcal{M}^*$ so that $\dist (M', \widehat{M})$ is minimised,
where $\dist (\cdot,\cdot)$ denotes the hamming distance between two vectors.
\item Output $\widehat{m} = e^{-1}(\widehat{M})$.
\end{enumerate}
\end{minipage}
  }
\end{center}
\end{figure}

\begin{theorem}
\label{t1.1}
Suppose that $\mathcal{P}$ is a proving algorithm for the \textsf{Basic  Scheme} for which
$\mathsf{succ}(\mathcal{P}) > 1 - d/(2n)$, where the
hamming distance of the set of encoded messages $\mathcal{M}^*$ is $d$.
Then the \Ext \ presented in Figure \ref{fig1} will always output $\widehat{m} = m$.
\end{theorem}

\begin{proof}
Let $M'$ be the $n$-tuple of responses computed by $\mathcal{P}$ and denote $\delta = \dist (M,M')$, where $M = e(m)$.
Denote $\epsilon = \mathsf{succ}(\mathcal{P})$. Then it is easy to see that $\epsilon = 1 - \delta/n$.
We want to prove that $\widehat{M} = M$.  We have that $\widehat{M}$ is a vector in $\mathcal{M}^*$ closest to $M'$.
 Since $M$ is a vector in $\mathcal{M}^*$ such that $\dist (M,M') = \delta$, it must be the case that $\dist (\widehat{M},M')
\leq \delta$. By the triangle inequality, we get
\[ \dist (M, \widehat{M}) \leq  \dist (M,M') + \dist (\widehat{M},M') \leq \delta + \delta = 2\delta.\]
However, \[ 2\delta = 2n(1 - \epsilon) < 2n\left( \frac{d}{2n} \right) = d.\]
Since $M$ and $\widehat{M}$ are vectors in $\mathcal{M}^*$ within distance $d$,
it follows that $M  = \widehat{M}$
and the \Ext \ outputs $m = e^{-1}(M)$, as desired.
\end{proof}

\subsection{General Keyless Challenge-Response Schemes}

In the simple protocol we analysed, a response to a challenge was just a particular
message block chosen from an encoded message $M$. We are interested in studying more complicated
protocols. For example, we might consider a response that consists of several message blocks from $M$,
or is computed as a function of one or more message blocks. In this section, we generalise
the preceding extraction process to handle arbitrary keyless challenge-response protocols.

In general, a challenge will be chosen from a specified \emph{challenge space} $\Gamma$,
and the response will be an element of a \emph{response space} $\Delta$.
The \emph{response function} $\rho : \mathcal{M}^* \times \Gamma \rightarrow \Delta $
computes the response
$r = \rho(M,c)$ given the encoded message $M$ and the challenge $c$.

For an encoded message $M \in \mathcal{M}^*$, we define the \emph{response vector}
\[r^M = ( \rho(M,c) : c \in \Gamma ) .\] That is, $r^M$ contains all the  responses
to all possible challenges for the encoded message $M$. Finally, define the \emph{response code}
(or more simply, the \emph{code}) of the
scheme to be
\[ \mathcal{R}^* = \{ r^M : M \in \mathcal{M}^* \} .\]
The codewords in $\mathcal{R}^*$ are just the response vectors that we defined above.

The \textsf{Generalised Scheme} is presented in Figure \ref{fig-generalised}.
Observe that $\mathcal{R}^* \subseteq \Delta^{\gamma}$, where $\gamma = |\Gamma|$.
We will assume that the mapping $M \mapsto r^M$ is an injection, and therefore
the hamming distance of $\mathcal{R}^*$ is greater than $0$ (in fact, we make this assumption for
all the schemes we consider in this paper).

\begin{figure}[tb]
\caption{\textsf{Generalised Scheme}}\label{fig-generalised}
\begin{center}
\fbox{
    \begin{minipage}{15cm}
    \begin{description}
\item [initialisation] \mbox{\quad} \\ Given a \emph{message} $m \in \mathcal{M}$, encode $m$ as
$e(m) = M \in \mathcal{M}^*$.
The \Ver \ gives $M$ to the
\Pro. The \Ver \ also generates a random challenge $c \in \Gamma$ and
stores $c$ and $\rho(M,c)$.

\item [challenge-response] \mbox{\quad} \\
The \Ver \ gives the challenge $c$ to the \Pro.
The \Pro \ responds with $r= \rho(M,c)$. The \Ver \ checks that the value
$r$ returned by the \Pro \ matches the stored value $\rho(M,c)$.
\end{description}
\end{minipage}
  }
\end{center}
\end{figure}

We now describe the generalised adversarial model.

\begin{enumerate}
\item The \Adv \  is given the response code $\mathcal{R}^*$.
\item The \Adv \   selects $m \in \mathcal{M}$.
\item The \Adv \   outputs a {deterministic} proving algorithm for $m$, denoted $\mathcal{P}$.
\end{enumerate}

The {\it success probability} of $\mathcal{P}$ is defined to be
\[\mathsf{succ}(\mathcal{P}) = \mathbf{Pr} [\mathcal{P}(c) = \rho(M,c)],\]
where $M = e(m)$ and this probability is computed over
a challenge $c \in \Gamma $ chosen uniformly at random.

Now we  construct an \Ext \ that will take as input a proving algorithm
$\mathcal{P}$ for some unknown message $m$. The \Ext \ will output a message
$\widehat{m} \in \mathcal{M}$. We say that the \POR \ scheme is {\it secure} provided that
$\widehat{m} = m$ whenever  the success probability of $\mathcal{P}$ is sufficiently close to $1$.
The \Ext \ for the \textsf{Generalised  Scheme} is presented in Figure \ref{fig2}.

\begin{figure}[tb]
\caption{\Ext \ for the \textsf{Generalised  Scheme}}\label{fig2}
\begin{center}
\fbox{
    \begin{minipage}{15cm}
    \begin{enumerate}
\item On input $\mathcal{P}$, compute the vector $R' = (r'_c : c \in \Gamma)$,
where $r'_c = \mathcal{P}(c)$ for all $c \in \Gamma$ (i.e., for every $c$, $r'_c$ is
the response computed by $\mathcal{P}$ when it is given
the challenge $c$).
\item Find $\widehat{M} \in \mathcal{M}^*$ so that $\dist (R', r^{\widehat{M}})$ is minimised.
\item Output $\widehat{m} = e^{-1}(\widehat{M})$.
\end{enumerate}
\end{minipage}
  }
\end{center}
\end{figure}

The following theorem relates the success probability of the extractor to the
hamming distance of the response code. Its proof is essentially identical to
the proof of Theorem \ref{t1.1}.

\begin{theorem}
\label{t1.2}
Suppose that $\mathcal{P}$ is a proving algorithm for the \textsf{Generalised    Scheme} for which
$\mathsf{succ}(\mathcal{P}) > 1 - d^*/(2\gamma)$, where the
hamming distance of the response code $\mathcal{R}^*$ is $d^*>0$.
Then the \Ext \ presented in Figure \ref{fig2} will always output $\widehat{m} = m$.
\end{theorem}

\begin{proof}
Let $R'$ be the $\gamma$-tuple of responses computed by $\mathcal{P}$ and denote $\delta = \dist (r^M,R')$,
where $M = e(m)$.
Denote $\epsilon = \mathsf{succ}(\mathcal{P})$. Then it is easy to see that $\epsilon = 1 - \delta/\gamma$.
We want to prove that $\widehat{M} = M$.
We have that $r^{\widehat{M}}$ is a codeword in $\mathcal{R}^*$ closest to $R'$.
 Since $M$ is a codeword such that $\dist (r^M,R') = \delta$, it must be the case that
 $\dist (r^{\widehat{M}},R')
\leq \delta$. By the triangle inequality, we get
\[ \dist (r^M, r^{\widehat{M}}) \leq  \dist (r^M,R') + \dist (r^{\widehat{M}},R') \leq \delta + \delta = 2\delta.\]
However, \[ 2\delta = 2\gamma(1 - \epsilon) < 2\gamma\left( \frac{d^*}{2\gamma} \right) = d^*.\]
Since $r^M$ and $r^{\widehat{M}}$ are codewords within distance $d^*$,
it follows that $M  = \widehat{M}$
and the \Ext \ outputs $m = e^{-1}(M)$, as desired.
\end{proof}

Observe that the efficacy of this extraction process depends on the \emph{relative distance} of
the response code $\mathcal{R}^*$, which equals $d^* / \gamma$.
In the next subsections, we look at some specific examples of
challenge-response protocols.

\subsection{Multiblock Challenge  Scheme}

We present a \POR \  scheme that we term the \textsf{Multiblock Challenge    Scheme} in Figure
\ref{fig-multiblock} (this is the same as the  ``Basic PoR Code Construction'' from \cite{DVW}).

\begin{figure}[tb]
\caption{\textsf{Multiblock Challenge   Scheme}}\label{fig-multiblock}
\begin{center}
\fbox{
    \begin{minipage}{15cm}
\begin{description}
\item [challenge] \mbox{\quad} \\ A challenge is a subset of $\ell$ indices
$J \subseteq \{1, \dots , n\}$.
Therefore,
$\Gamma = \{ J \subseteq \{1, \dots , n\}, |J| = \ell\}$
 and $\gamma = \binom{n}{\ell}$.
\item [response] \mbox{\quad} \\
Given the challenge $J = \{i_1, \dots , i_{\ell}\}$ where
$1 \leq i_1 < \cdots < i_{\ell} \leq n$,
the correct response is the $\ell$-tuple \[ \rho(M,J) = (m_{i_1}, \dots , m_{i_{\ell}}).\]
{Suppose the \Ver \ receives a 
response $(r_1, \dots , r_{\ell})$. He then checks that 
$r_j = m_{i_j}$ for $1 \leq j \leq \ell$.}
\end{description}
In this scheme, we have $\Delta = (\eff_q)^{\ell}$.
\end{minipage}
  }
\end{center}
\end{figure}

\begin{lemma}
\label{l1.3}
Suppose that the hamming distance of $\mathcal{M}^*$ is $d$.
Then the hamming distance of the response code of the \textsf{Multiblock Challenge   Scheme}
is $d^* = \binom{n}{\ell} - \binom{n-d}{\ell}$.
\end{lemma}

\begin{proof}
Suppose that $M, M' \in \mathcal{M}^*$ and $M \neq M'$.
Denote $\dist (M,M') = \delta$,
$M = (m_1, \dots , m_n)$ and $M' = (m'_1, \dots , m'_n)$.
It is easy to see that $r^M_J = r^{M'}_J$ if and only if
$J \subseteq \{ i : m_i = m'_i\}$. From this, it is immediate that
$\dist (r^M, r^{M'}) = \binom{n}{\ell} - \binom{n-\delta}{\ell}$.
The desired result follows because $\delta \geq d$.
\end{proof}

\begin{theorem}
\label{c1.4}
Suppose that $\mathcal{P}$ is a proving algorithm for the \textsf{Multiblock Challenge   Scheme} for which
\[ \mathsf{succ}(\mathcal{P}) > \frac{1}{2} + \frac{\binom{n-d}{\ell}}{2\binom{n}{\ell}},\]
where the hamming distance of $\mathcal{M}^*$ is $d$.
Then the \Ext \ presented in Figure \ref{fig2} will always output $\widehat{m} = m$.
\end{theorem}

\begin{proof}This is an immediate consequence of Theorem \ref{t1.2} and Lemma \ref{l1.3},
once we verify that
\[ 1 - \frac{d^*}{2\gamma} = 1 - \frac{\binom{n}{\ell} - \binom{n-d}{\ell}}{2\binom{n}{\ell}}
= \frac{1}{2} + \frac{\binom{n-d}{\ell}}{2\binom{n}{\ell}}.\]
\end{proof}

\noindent{\bf Remark:} When we set $\ell = 1$ in Corollary
\ref{c1.4}, we obtain Theorem \ref{t1.1}.

\subsection{Linear Combination   Scheme}

In this subsection, we consider the \textsf{Linear Combination  Scheme}, in which
a response consists of a specified linear combination of  message blocks
(this could perhaps be thought of as a keyless analogue of the Shacham-Waters scheme \cite{SW}).
The scheme is presented in Figure \ref{fig-linear}.
We will study two versions of the scheme:
\begin{description}
\item[Version 1] \mbox{\quad} \\ Here, the challenge $V$  is
any non-zero vector in  $(\eff_q)^n$, so $\gamma = q^n-1$.
\item[Version 2] \mbox{\quad} \\ In this version of the scheme, the challenge $V$ is
a vector in  $(\eff_q)^n$ having hamming weight equal to $\ell$, so $\gamma = \binom{n}{\ell}(q-1)^{\ell}$.
\end{description}


\begin{figure}[tb]
\caption{\textsf{Linear Combination   Scheme}}\label{fig-linear}
\begin{center}
\fbox{
    \begin{minipage}{15cm}
\begin{description}
\item [challenge] \mbox{\quad} \\ A challenge is an
$n$-tuple $V = (v_1, \dots , v_n) \in (\eff_q)^n$.
\item [response] \mbox{\quad} \\
Given the challenge $V = (v_1, \dots , v_n)$,
the correct response is \[ \rho(M,V) = V \cdot M = \sum_{i=1}^n v_im_i,\]
where $M = (m_1, \dots , m_n)$ and the computation is performed in $\eff_q$.
{Suppose the \Ver \ receives a 
response $r \in \eff_q$. He then checks that 
$r = V \cdot M$.}
\end{description}
In this scheme,  $\Delta = \eff_q$.
\end{minipage}
  }
\end{center}
\end{figure}

\subsubsection{Analysis of Version 1}

\begin{lemma}
\label{l1.6v1}
Suppose that the hamming distance of $\mathcal{M}^*$ is $d$.
The hamming distance of the response code of version 1 of the \textsf{Linear Combination   Scheme}
is
$d^* =  q^n - q^{n-1} - 1$.
\end{lemma}

\begin{proof}
Suppose $M, M' \in \mathcal{M}^*$ and $M \neq M'$.
Then $V \cdot M  = V \cdot M'$ if and only if $V \cdot (M - M') = 0$.
Since $M \neq M'$, there are $q^{n-1}$ solutions for $V$, given $M$ and $M'$.
There are $q^n - 1$ choices for $V$, so the desired result follows.
\end{proof}

We observe that $d^*$ is independent of $d$ (the hamming distance of $\mathcal{M}^*$)
in version 1 of the  scheme.

\begin{theorem}
Suppose that $\mathcal{P}$ is a proving algorithm for
version 1 of the \textsf{Linear Combination   Scheme} for which
\[ \mathsf{succ}(\mathcal{P}) > \frac{1}{2} + \frac{1}{2} \left( \frac{q^{n-1}}{q^n-1} \right).\]
Then the \Ext \ presented in Figure \ref{fig2} will always output $\widehat{m} = m$.
\end{theorem}

\subsubsection{Analysis of Version 2}

\begin{lemma}
\label{l1.5}
Suppose that $r \geq 1$ and $X \in (\eff_q)^r$ has hamming weight equal to $r$.
Then the number of solutions {$V \in (\eff_q)^r$} to the equation $V \cdot X = 0$ in which
$V$ has hamming weight equal to $r$, which we denote by $a_r$, is given by the formula
\begin{equation}
\label{eq1}
a_r = \frac{q-1}{q}((q-1)^{r-1} - (-1)^{r-1}).
\end{equation}
\end{lemma}

\begin{proof}
We prove the result by induction on $r$. When $r =1$, there
are no solutions, so $a_1 = 0$, agreeing with (\ref{eq1}).
Now assume that (\ref{eq1}) gives the number of solutions for $r= s-1$,
and consider $r=s$. Let $X = (x_1, \dots  , x_s)$ and define
$X' = (x_1, \dots  , x_{s-1})$. By induction, the number of solutions to the equation
$V' \cdot X' = 0$ in which $V'$ has hamming weight $s-1$ is $a_{s-1}$.
Each of these solutions $V'$ can be extended to a solution of the equation
$V \cdot X = 0$ by setting $v_s =0$; in each case, the resulting $V$ has hamming weight
equal to $s-1$. However, any other vector $V'$ of weight $s-1$ can be extended to
a solution of the equation
$V \cdot X = 0$ which has hamming weight equal to $s$.
Therefore, we have
\begin{eqnarray*}
a_s &=& (q-1)^{s-1} - a_{s-1}\\
&=&
(q-1)^{s-1} - \left(\frac{q-1}{q}((q-1)^{s-2} - (-1)^{s-2})\right)\\
&=& (q-1)^{s-1} \left( 1 - \frac{1}{q} \right) +  \frac{q-1}{q}(-1)^{s-2} \\
&=&  \frac{q-1}{q}((q-1)^{s-1} - (-1)^{s-1}).
\end{eqnarray*}
\end{proof}

\noindent{\em Remark.}
An alternative way to prove (\ref{eq1}) is to observe that
$a_r$ is equal to the number of codewords of weight $r$ in an
MDS code having length $r$, dimension $r-1$ and distance $2$. Then (\ref{eq1}) can be
derived from well-known formulas for the weight distribution of an MDS code.
For example, in \cite[Ch.\ 6, Theorem 6]{MS}, it is shown that the
number of codewords of weight $w$ in an MDS code of length $n$, dimension
$k$ and distance $d=n-k+1$ over $\eff_q$ is
\[ a_w = \binom{n}{w}(q-1)\sum_{j=0}^{w-d} (-1)^j \binom{w-1}{j} q^{w-d-1} .\]
If we substitute $n=w=r$, $d=2$ into this formula, we get
\begin{eqnarray*}
a_r &=&  \binom{r}{r}(q-1)\sum_{j=0}^{r-2} (-1)^j \binom{r-1}{j} q^{r-2-j}\\
&=& \frac{q-1}{q} \sum_{j=0}^{r-2} \binom{r-1}{j} (-1)^j  q^{r-1-j}\\
&=& \frac{q-1}{q} \left( \sum_{j=0}^{r-1} \binom{r-1}{j} (-1)^j  q^{r-1-j}
- (-1)^{r-1} \right) \\
&=& \frac{q-1}{q} \left( (q-1)^{r-1} - (-1)^{r-1} \right) ,
\end{eqnarray*}
 which agrees with (\ref{eq1}), as proven in Lemma \ref{l1.5}.

 \medskip

 We will now compute the distance $d^*$ of the response code $\mathcal{M}^*$.
Here is a lemma that will be of use in computing $d^*$.

\begin{lemma}
\label{l1.6}
Suppose that $M, M' \in \mathcal{M}^*$ and $M \neq M'$.  Denote $\delta = \dist (M,M')$.
Let $r^M$ and $r^{M'}$ be the corresponding vectors in the response code of
Version 2 of the \textsf{Linear Combination   Scheme}.
Then
\begin{equation}
\label{eq2}
\dist (r^M, r^{M'}) =  (q-1)^{\ell} \left( \binom{n}{\ell}  - \binom{n-\delta}{\ell} \right)  -
\sum_{w \geq 1} \binom{\delta}{w} \binom{n-\delta}{\ell - w} (q-1)^{\ell-w}a_w,
\end{equation}
where the $a_w$'s are given by (\ref{eq1}).
\end{lemma}

\begin{proof}
Suppose that $M, M' \in \mathcal{M}^*$ and $M \neq M'$.
Denote
\[ M = (m_1, \dots , m_n) \text{ and } M' = (m'_1, \dots , m'_n),\]
and let $\delta = \dist (M,M')$.
Let \[ J = \{ i : m_i = m'_i\} \text{ and }J' = \{1, \dots , n\} \setminus J.\]
Observe that $|J| = n-\delta$ and $|J'| = \delta$.
For any $V = (v_1, \dots, v_n)$ having hamming weight equal to $\ell$, define \[J_V =
\{j \in J : v_j \neq 0\} \text{ and }J'_V =
\{j \in J' : v_j \neq 0\}.\]
Denote $w = |J'_V|$; then $|J_V| = \ell - w$.

Suppose $w \geq 1$. Then, given $J_V$ and $J'_V$, the number of solutions to
the equation $V \cdot M = V \cdot M'$ is precisely $(q-1)^{\ell-w}a_w$.
When $w=0$, the number of solutions is $(q-1)^{\ell}$.
Summing over $w$, and considering all possible choices for $J_V$ and $J'_V$,
we see that the total number of solutions to
the equation $V \cdot M = V \cdot M'$ is
\begin{equation}
\label{lc.eq}   \binom{n-\delta}{\ell} (q-1)^{\ell} +
\sum_{w \geq 1} \binom{\delta}{w} \binom{n-\delta}{\ell - w} (q-1)^{\ell-w}a_w.
\end{equation}
The desired result follows.
\end{proof}

We can obtain a very accurate estimate for $d^*$ by observing that
\[ a_w \approx \frac{(q-1)^w}{q}\] is a very accurate approximation.
After making this substitution, it is easy to see that
the quantity (\ref{lc.eq}) is  minimised when $\delta = d$. This
minimises (\ref{eq2}), so we obtain
\begin{equation}
\label{lc.eq2}
 d^* \approx \binom{n}{\ell} (q-1)^{\ell} - \binom{n-d}{\ell} (q-1)^{\ell} -
\sum_{w \geq 1} \binom{d}{w} \binom{n-d}{\ell - w} \frac{(q-1)^{\ell}}{q}.
\end{equation}
We have
\begin{eqnarray*}
 \sum_{w \geq 1} \binom{d}{w} \binom{n-d}{\ell - w} \frac{(q-1)^{\ell}}{q}
& = & \frac{(q-1)^{\ell}}{q} \sum_{w \geq 1} \binom{d}{w} \binom{n-d}{\ell - w} \\
& = & \frac{(q-1)^{\ell}}{q} \left( \binom{n}{\ell} -  \binom{n-d}{\ell} \right).
\end{eqnarray*}
Therefore, from (\ref{lc.eq2}), we get
\begin{align}
 d^* & \approx  (q-1)^{\ell} \left( \binom{n}{\ell}  - \binom{n-d}{\ell} \right)
- \frac{(q-1)^{\ell}}{q} \left( \binom{n}{\ell} -  \binom{n-d}{\ell} \right)\nonumber \\
\label{estimate} & =  \frac{(q-1)^{\ell + 1}}{q} \left( \binom{n}{\ell} -  \binom{n-d}{\ell} \right).
\end{align}

The following theorem uses the estimated value for $d^*$ derived in
(\ref{estimate}).

\begin{theorem}
\label{LCv2.thm}
Suppose that $\mathcal{P}$ is a proving algorithm for
version 2 of the \textsf{Linear Combination  Scheme} for which
\[ \mathsf{succ}(\mathcal{P}) > \frac{1}{2} + \frac{1}{2} \left(
\frac{1}{q} + \frac{(q-1)\binom{n-d}{\ell}}{q\binom{n}{\ell}} \right),\]
where the hamming distance of $\mathcal{M}^*$ is $d$.
Then the \Ext \ presented in Figure \ref{fig2} will always output $\widehat{m} = m$.
\end{theorem}

\section{Analysis of a Keyed Scheme: the Shacham-Waters Scheme}
\label{keyed.sec}

The \textsf{Shacham-Waters Scheme} \cite{SW} is a {\it keyed proof-of-retrievability scheme}.
This means that the
\Ver \ has a secret key that is not provided to the \Pro. This key is used to verify responses
in the challenge-response protocol, and it is also provided to
an extraction algorithm as input.
The use of a key permits an arbitrary number of challenges to be verified, without
the \Ver \ having to precompute the responses.

We discuss a variation of the \textsf{Shacham-Waters Scheme} \cite{SW},
modified to fit the unconditional security
setting. The main change is that the vector $(\beta_1, \dots , \beta_n)$ (which comprises
part of the key) is completely random,
rather than being generated by a pseudorandom function (i.e., a PRF).
This scheme, presented in Figure \ref{fig-SW}, is termed  the \textsf{Modified Shacham-Waters Scheme}.

\begin{figure}[tb]
\caption{\textsf{Modified Shacham-Waters Scheme}}\label{fig-SW}
\begin{center}
\fbox{
    \begin{minipage}{15cm}
\begin{itemize}
\item The {\it key} $K$ consists of $\alpha \in \eff_q$ and $B = (\beta_1, \dots , \beta_n) \in (\eff_q)^n$.
$K$ is retained by the \Ver.
\item The {\it encoded message} is $M = (m_1, \dots , m_n) \in (\eff_q)^n$.
\item The {\it tag} is $S = (\sigma_1, \dots , \sigma_n) \in (\eff_q)^n$, where
$S$ is computed using the following (vector) equation in $\eff_q$:
\begin{equation}
\label{tag.eq} S = B + \alpha M.
\end{equation}
The message $M$ and the tag $S$ are given to the \Pro.
\item A {\it challenge} is a
vector 
 $V=(v_1, \dots , v_n) \in (\eff_q)^n$.
\item The {\it response} consists of $(\mu,\tau) \in (\eff_q)^2$, where
the following computations are performed in $\eff_q$:
\begin{equation}
\label{mu.eq} \mu = V \cdot M
\end{equation}
and
\begin{equation}
\label{tau.eq} \tau = V \cdot S.
\end{equation}
\item The response $(\mu,\tau)$ is {\it verified} by checking that the following
condition holds in $\eff_q$:
\begin{equation}
\label{ver.eq} \tau \stackrel{?}{=} \alpha \mu + V \cdot B.
\end{equation}
\end{itemize}
\end{minipage}
  }
\end{center}
\end{figure}

As we did  with the \textsf{Linear Combination    Scheme}, we will study two versions of the scheme.
In the first version, the challenge $V$ is
any non-zero vector in  $(\eff_q)^n$, so $\gamma = q^n-1$.
In the second version, the challenge $V$ is
a vector in  $(\eff_q)^n$ having hamming weight equal to $\ell$, so $\gamma = \binom{n}{\ell}(q-1)^{\ell}$.
In both versions, $\Delta = (\eff_q)^2$.

The information held by the various parties in the scheme is
summarised as follows:

\begin{center}
\begin{tabular}{c|c|c}
\Ver & \Pro & \Ext \\ \hline
$K= (\alpha, B)$ & $M, S$ & $K, \mathcal{P}$
\end{tabular}
\end{center}

We first observe that, from the point of view of the \Pro, there are $q$ possible
keys.
\begin{lemma}
\label{alpha.lem}
Given $M$ and $S$, the \Pro \ can restrict the set of possible keys $(\alpha, B)$ to
\[ \Pos(M,S) = \{ (\alpha_0, S - \alpha_0 M) : \alpha_0 \in \eff_q\} .\]
\end{lemma}

\begin{proof}
Suppose that $\alpha = \alpha_0$. Then
equation (\ref{tag.eq}) implies that $B = S - \alpha_0 M$.
\end{proof}

We will define a response $(\mu,\tau)$ to be {\it acceptable} if (\ref{ver.eq}) is satisfied.
A response is {\it authentic} if it was created using equations (\ref{mu.eq}) and (\ref{tau.eq}).
Note that an authentic response will be acceptable for every key $K  \in \Pos(M,S)$.
In the case of an acceptable (but perhaps not authentic) response, we have the following useful
lemma.

\begin{lemma}
\label{acceptable.lem}
Suppose that a response $(\mu,\tau)$ to a challenge $V$ for a message $M$ is acceptable for
more than one key in $\Pos(M,S)$. Then $(\mu,\tau)$ is authentic.
\end{lemma}

\begin{proof} Suppose $K_1,K_2 \in \Pos(M,S)$, where
$K_1 = (\alpha_1,B_1)$, $K_2 = (\alpha_2,B_2)$ and $\alpha_1 \neq \alpha_2$.
We have $B_1 = S - \alpha_1 M$ and $B_2 = S - \alpha_2 M$.
Now consider a response $(\mu,\tau)$ to a challenge $V$  that is acceptable
for both of the keys $K_1$ and $K_2$. Then
\[ \tau = \alpha_1 \mu + V \cdot B_1 = \alpha_2 \mu + V \cdot B_2.\]
Therefore,
\[ (\alpha_1 - \alpha_2)\mu + V \cdot (B_1 - B_2) = 0.\]
However, we have $B_1 - B_2 = (\alpha_2 - \alpha_1)M$, so
\[ (\alpha_1 - \alpha_2)(\mu - V \cdot M) = 0.\]
We have $\alpha_1 \neq \alpha_2$, so it follows that
$\mu = V \cdot M$.
Then we obtain
\begin{eqnarray*}
\tau & = & \alpha_1 \mu + V \cdot B_1 \\
&= & \alpha_1 V \cdot M + V \cdot (S - \alpha_1 M) \\
& = & V \cdot S.
\end{eqnarray*}
Therefore the response $(\mu,\tau)$ is authentic.
\end{proof}

The verification condition (\ref{ver.eq}) depends on the key, but not on the
message $M$.
It follows that the \Pro \ can create acceptable non-authentic responses if he knows the value of $\alpha$
(which in turn uniquely determines the key, as shown in Lemma \ref{alpha.lem}).
This leads to the following attack.

\begin{theorem}
If the \Pro \ has access to a verification oracle, then the \textsf{Modified Shacham-Waters Scheme} is not unconditionally secure.
\end{theorem}

\begin{proof} 
For every key $K \in \Pos(M,S)$, the \Pro \ can create a response $(\mu,\tau)$
to a challenge $V$ that will be acceptable if and only if $K$ is the
actual key (this follows from Lemma \ref{acceptable.lem}).
The \Pro \ can check the validity of these responses by accessing the verification
oracle. As soon as one of these responses is accepted by the verification oracle, the
\Pro \ knows the correct value of the key. Hence the \Pro \ can now create a proving algorithm
$\mathcal{P}$ that will output
acceptable but non-authentic responses.
This algorithm $\mathcal{P}$ will not allow the
correct message to be extracted.

In more detail, after the \Pro \ has determined
the key $K = (\alpha,B)$, he chooses an arbitrary (encoded) message $M' \neq M$ and constructs $\mathcal{P}$
as follows:
\begin{enumerate}
\item Given a challenge $V$, define $\mu = V \cdot M'$.
\item Then define
$\tau = \alpha \mu + V \cdot B$.
\end{enumerate}
Suppose the \Ext \ is run on $\mathcal{P}$. It is easy to
see that $\dist (R', r^{M'}) = 0$, so the \Ext \ will  compute $\widehat{M} = M'$,
which is incorrect.
\end{proof}



Even in the absence of a verification oracle, the \Pro \ can guess the correct value of
$\alpha$, which he can do successfully with probability $1/q$. If he correctly guesses $\alpha$, he can
create a non-extractable proving algorithm. This implies that it is not possible to prove
a theorem stating that {\it any} proving algorithm yields an extractor. However, we can prove meaningful
reductions if we define the success probability of a proving algorithm
to be the \emph{average} success probability
over the $q$ possible keys that are consistent with the information given to the \Pro.

Suppose $\mathcal{P}$
is a (deterministic) proving algorithm for a message $M = (m_1, \dots , m_n) \in (\eff_q)^n$.
For each challenge $V$ and for every key $K = (\alpha,B) \in \Pos(M,S)$,
define $\chi(V,K) = 1$ if $\mathcal{P}$ returns an acceptable response for the key $K$ given the
challenge $V$, and define $\chi(V,K) = 1$, otherwise.

Since there are $\gamma q$ choices for the pair $(V,K)$,
we define the {\it average success probability} $\mathsf{succ}_{\mathsf{avg}}(\mathcal{P})$  to be
\begin{equation}
\label{avg-succ} \mathsf{succ}_{\mathsf{avg}}(\mathcal{P})
= \frac{\displaystyle{\sum_{V\in \Gamma,K \in \Pos(M,S)} \chi(V,K)}}{\gamma q} .
\end{equation}

\begin{lemma}
\label{succ.lem}
Suppose there are $D$ challenges $V$ for which $\mathcal{P}$ returns an authentic
response, and hence there are $C = \gamma -D $ challenges for which $\mathcal{P}$  returns a response
that is not authentic.  Then
\begin{equation}
\label{succ.eq} \mathsf{succ}_{\mathsf{avg}}(\mathcal{P}) \leq  1 - \frac{C(q-1)}{\gamma q}.
\end{equation}
\end{lemma}

\begin{proof}
If $V$ is a challenge for which $\mathcal{P}$ returns an authentic
response, then $\chi(V,K) = 1$ for every $K$.
If $V$ is a challenge for which $\mathcal{P}$ does not return an authentic
response, then Lemma \ref{acceptable.lem} implies that $\chi(V,K) = 1$ for at most one $K$.
Therefore,
\[ \sum_{V\in \Gamma,K \in \Pos(M,S)} \chi(V,K)  \leq C + qD = \gamma q - C(q-1).\]
The desired result now follows from (\ref{avg-succ}).
\end{proof}

Let's now turn our attention to the response code. Even though
a response is an ordered pair $(\mu,\tau)$, it suffices to consider only the values
of $\mu$ in the extraction process
(this follows because $\mu$, $K$ and $V$ uniquely determine $\tau$; see (\ref{ver.eq})). So we will define the
response vector for a message $M$ to be
$r^M = (M \cdot V : V \in \Gamma)$.
Observe that this response vector is identical to the
response vector in the \textsf{Linear Combination   Scheme}.

\begin{lemma}
\label{L2.5}
Suppose that $\mathcal{P}$ is a proving algorithm for the
message $M$ in the \textsf{Modified Shacham-Waters Scheme}.
Let $r^M = (M \cdot V : V \in \Gamma)$ and
let $R'$ be the $\gamma$-tuple of responses computed by $\mathcal{P}$.
Then \[ \dist (r^M,R') \leq \frac{(1 - \mathsf{succ}_{\mathsf{avg}}(\mathcal{P})) \gamma q}{q-1}.\]
\end{lemma}

\begin{proof}
Define $C$ as in Lemma \ref{succ.lem}. Since a co-ordinate of $r^M$ differs from the corresponding
co-ordinate of $R'$ only when the response is non-authentic, it follows that
$\dist (r^M,R') \leq C$. Equation (\ref{succ.eq})
implies that
\[ C \leq \frac{(1 - \mathsf{succ}_{\mathsf{avg}}(\mathcal{P})) \gamma q}{q-1},\]
from which the stated result follows.
\end{proof}

\begin{figure}[tb]
\caption{\Ext \ for the \textsf{Modified Shacham-Waters Scheme}}\label{fig8}
\begin{center}
\fbox{
    \begin{minipage}{15cm}
    \begin{enumerate}
\item On input $\mathcal{P}$, compute the vector $R' = (\mu_V : V \in \Gamma)$,
where $(\mu_V,\tau_V) = \mathcal{P}(V)$ for all $V \in \Gamma$.
\item Find $\widehat{M} \in \mathcal{M}^*$ so that $\dist (R', r^{\widehat{M}})$ is minimised.
\item Output $\widehat{m} = e^{-1}(\widehat{M})$.
\end{enumerate}
\end{minipage}
  }
\end{center}
\end{figure}

\begin{theorem}
\label{SWsucc.thm}
Suppose that
\begin{equation}
\label{eq2.6} \mathsf{succ}_{\mathsf{avg}}(\mathcal{P}) > 1 - \frac{d^*(q-1)}{2 \gamma q},
\end{equation}
where $d^*$ is given by equation (\ref{eq2}).
Then the \Ext \ presented in Figure \ref{fig8} will always output $\widehat{m} = m$.
\end{theorem}

\begin{proof}
Denote $\epsilon = \mathsf{succ}_{\mathsf{avg}}(\mathcal{P})$,
let $R'$ be the $\gamma$-tuple of responses computed by $\mathcal{P}$, and denote $\delta = \dist (r^M,R')$,
where $M = e(m)$.  We showed in Lemma \ref{L2.5} that
\[ \delta \leq \frac{(1 - \epsilon) \gamma q}{q-1}.\]

We want to prove that $\widehat{M} = M$.
We have that $r^{\widehat{M}}$ is a codeword in $\mathcal{R}^*$ closest to $R'$.
 Since $M$ is a codeword such that $\dist (r^M,R') = \delta$, it must be the case that
 $\dist (r^{\widehat{M}},R') \leq \delta$.
 By the triangle inequality, we get
\[ \dist (r^M, r^{\widehat{M}}) \leq  \dist (r^M,R') + \dist (r^{\widehat{M}},R') \leq \delta + \delta = 2\delta.\]
However, \[ 2\delta \leq \frac{2(1 - \epsilon) \gamma q}{q-1} <  d^*,\]
where the last inequality follows from (\ref{eq2.6}).
Since $r^M$ and $r^{\widehat{M}}$ are codewords within distance $d^*$ (which is the distance of
the response code),
it follows that $M  = \widehat{M}$
and the \Ext \ outputs $m = e^{-1}(M)$, as desired.
\end{proof}

\section{Numerical Computations and Estimates}
\label{numer.sec}

We have provided sufficient conditions for extraction to succeed for
several \POR \ schemes, based on
the success probability of the proving algorithm. Here look a bit more closely
at these numerical conditions and provide some useful comparisons
and estimates for the different schemes we have studied.

We will consider three schemes. We have the following observations, which
can be verified in a straightforward manner:
\begin{enumerate}
\item  For the \textsf{Multiblock Challenge   Scheme}, extraction will succeed if
\[ \mathsf{succ}(\mathcal{P}) > S_0 = \frac{1}{2} + \frac{\binom{n-d}{\ell}}{2\binom{n}{\ell}}.\]
This is stated in Theorem \ref{c1.4}.
\item For the  \textsf{Linear Combination  Scheme (Version 2)},
extraction will succeed if
\[ \mathsf{succ}(\mathcal{P}) > S_1 = \left( \frac{q-1}{q}\right) S_0 + \frac{1}{q}.\]
This follows from Theorem \ref{LCv2.thm}.
\item For the  \textsf{Modified Shacham-Waters Scheme},
extraction will succeed if
\[ \mathsf{succ}_{\mathsf{avg}}(\mathcal{P}) > S_2 = \left( \frac{q-1}{q}\right)^2 S_0 + \frac{2}{q} - \frac{1}{q^2}.\]
This follows from Theorem \ref{SWsucc.thm}, using the estimate for $d^*$ given in (\ref{estimate}).
\end{enumerate}
It is clear that $S_0,S_1$ and $S_2$ are extremely close for any reasonable value of $q$
(such as $q \geq 2^{32}$, for example). Therefore we will confine our subsequent analysis  to $S_0$
and state our results in terms of the \textsf{Multiblock Challenge    Scheme}.
The formula for $S_0$ is relatively simple, but it is complicated somewhat by the binomial
coefficients. Therefore, it may be useful to define an estimate that does not involve binomial coefficients.

\begin{table}[htbp]
\begin{center}
{\small
\begin{tabular}{|rrl r r || rrl r r|}
\hline
$\ell$ & $d$ & $\mathsf{succ}(\mathcal{P})$ & $n$ (Thm.\ \ref{c1.4}) & $n$ (Thm.\ \ref{estimate.thm}) & $\ell$ & $d$ & $\mathsf{succ}(\mathcal{P})$ & $n$ (Thm.\ \ref{c1.4}) & $n$ (Thm.\ \ref{estimate.thm})\\ \hline
  10000 &  10000 & 0.6 &    $62143493  $ &62133493		&   10000 &  1000 & 0.6 &   $6218850  $& 6213349 \\
& & 0.7 &                  $109145666  $ 	&109135666		&& & 0.7		&                  $10919066   $ & 10913566\\
& & 0.8 &                  $195771518  $  	&195761518		&& & 0.8 	&                  $19581651   $ & 19576151\\
& & 0.9 &	                $448152011   $ 	&448142011			& && 0.9 		&$44819701   $& 44814201\\
& & 0.99 &	       $4949841645$ 	&4949831645		& && 0.99 	          &$494988664   $ & 494983164\\
              \hline
  1000 &  10000 & 0.6 &    $6218850  $&6213349  &  1000 &  1000 & 0.6 & $622334    $ & 6213349\\
& & 0.7 &                        $10919066   $&10913567 && & 0.7 &                        $1092356    $ & 10913567\\
& & 0.8 &                        $19581651   $&19576152 && & 0.8 &                        $1958614    $ & 19576152 \\
& & 0.9 &                        $44819700   $ &44814201&& & 0.9 &                        $4482419    $ & 4481420\\
& & 0.99 &                      $494988664  $&494983164 && & 0.99 &                      $49499315    $ & 49498316\\

 	\hline
  100 &  10000 & 0.6 &    	$ 626398  $&621334 &  100 &  1000 & 0.6 & $ 62684$ & 62133\\
& & 0.7 &                      $1096413  $ &1091357 && & 0.7 &                      $109685   $ &109135 \\
& & 0.8 &                      $1962669   $ &1957615&& & 0.8 &                      $196311   $ & 195761\\
& & 0.9 &                      $4486471   $ &4481420 && & 0.9 &                      $448691   $ & 448142 \\
& & 0.99 &                    $49503366   $ &49498316&& & 0.99 &                    $4950381   $ & 4949831\\
	\hline
 50 &  10000 & 0.6 &    	$  315719  $&310667 &50 &  1000 & 0.6 &    	$31594$ & 31068 \\
& & 0.7 &                      $550718  $&5456783	&& & 0.7 &                      $55093$ & 545678 \\
& & 0.8 &                      $983840   $ &9788076	&& & 0.8 &                      $98406$& 978807\\
& & 0.9 &                      $2245736   $ &2240710	&& & 0.9 &                      $224599  $ & 224071 \\
& & 0.99 &                    $24754183  $ &24749158 && & 0.99 &                    $2475440   $ & 2474916 \\
	\hline  

   10000 &  100 & 0.6 &    $626398   $&621334 &    10000 &  10 & 0.6 & 67272& 62133 \\
& & 0.7 &                  $1096413   $ &1091356	&& & 0.7 &                              $ 114216   $ & 109136\\
& & 0.8 &                  $1962668   $&1957615	&& & 0.8 &                               $200808   $ & 195761\\
& & 0.9 &                  $4486471   $ &4481420	& && 0.9 &                               $453165   $ & 448142\\
& & 0.99 &                $4950336   $&49498316	&&& 0.99 &                              $ 4954838   $& 4949832 \\
	\hline
  1000 &  100 & 0.6 & $62684$ &62133&   1000 &  10 & 0.6 &6731& 6213\\
& & 0.7 &                  $109685   $&109135&& & 0.7 &11425 &10914\\
& & 0.8 &                  $196311    $&195761& && 0.8 &20084 & 19576\\
& & 0.9 &                  $448692   $ &448142& && 0.9 &45320 & 44814 \\
& & 0.99 &               $4950381   $ &4949831& && 0.99 &$495488  $ & 494983 \\
	\hline
  100 &  100 & 0.6 &6313 &6213&  100 &  10 & 0.6 & 677 & 621\\
& & 0.7 & 11013&10913	&& & 0.7 & 1146 & 1091\\
& & 0.8 & 19675&19576	&& & 0.8 & 2012 & 1958\\
& & 0.9 & 44913&44814	& && 0.9 & 4536 & 4481 \\
& & 0.99 &$495082  $ &494983	&& & 0.99 & 49552 & 49498 \\
	\hline
 50 &  100 & 0.6 & 3181&3106	&  50 &  10 & 0.6 & 341& 311\\
& & 0.7 & 5531&5456	&& & 0.7 &576& 546 \\
& & 0.8 & 9862&9788	&& & 0.8 & 1009& 979 \\
& & 0.9 & 22481 &22407 && & 0.0 & 2270& 2240 \\
& & 0.99 & $247565  $ &247492 & & & 0.99 & 24779 & 24749 \\
\hline
\end{tabular}
\caption{Values of $n$  which the \Ext \ will always succeed.}
\label{table:n}
}
\end{center}
\end{table}

\begin{theorem}
\label{estimate.thm}
Denote $\epsilon = 1 - \mathsf{succ}(\mathcal{P})$.
Suppose that the following inequality holds in the \textsf{Multiblock Challenge    Scheme}:
\begin{equation}
\label{succ-estimate.eq}
\frac{\ell d}{n} > \ln \left( \frac{1}{1-2\epsilon} \right).
\end{equation}
Then the \Ext \ will always succeed.
\end{theorem}

\begin{proof}
   From (\ref{succ-estimate.eq}), we obtain
\[ \ln \left( 1-2\epsilon \right) > - \frac{\ell d}{n}.\]
Now $-x > \ln(1-x)$ for $0 < x < 1$,  we obtain
\[ \ln \left( 1-2\epsilon \right) > \ell \ln \left( 1 - \frac{d}{n} \right) .\]
Exponentiating both sides of this inequality, we have
\[ 1-2\epsilon > \left( 1 - \frac{d}{n} \right)^{\ell}.\]
It is easy to prove that
\[ \left( 1 - \frac{d}{n} \right)^{\ell} > \frac{\binom{n-d}{\ell}}{\binom{n}{\ell}},\]
so it follows that
\[ 1-2\epsilon > \frac{\binom{n-d}{\ell}}{\binom{n}{\ell}}.\]
   From this, we obtain
\[ \mathsf{succ}(\mathcal{P}) > \frac{1}{2} + \frac{\binom{n-d}{\ell}}{2\binom{n}{\ell}},\]
and hence the \Ext \ will always succeed.
\end{proof}

Table \ref{table:n} presents values of $n$ (the length of an encoded message), for different values of $\ell$ (the hamming weight of the challenge), $d$ (the distance of the code) and the success probability of the proving algorithm,  such that the \Ext \ is guaranteed to succeed. We tabulate the value of $n$ as specified
by Theorems \ref{c1.4} and \ref{estimate.thm}. We see, for a wide range of parameters, that
the estimate obtained in Theorem \ref{estimate.thm} is very close to the earlier value computed in
Theorem \ref{c1.4}.


\section{Estimating the Success Probability of a Prover}
\label{confidence.sec}

\subsection{Hypothesis Testing}\label{subsec:hypothesis}
The essential purpose of a \POR \ scheme is to assure the user that their file is
indeed being stored correctly, { i.e.}, in such a manner that the user
can recover the entire file if desired.  We have considered several schemes
for testing whether this is the case, but in order to use these schemes
appropriately it is also necessary to pay attention to how we interpret the results of these tests.  Theorem~\ref{t1.1} tells us that extraction is possible for
the {\sf Basic Scheme} whenever $\succp$ is at least
$({n-\lfloor\frac{d}{2}\rfloor+1})/{n}$, hence the
information we would like to obtain from using the {\sf Basic Scheme} is
whether or not $\succp$ exceeds the necessary threshold.  Similarly, for the \textsf{Multiblock Challenge    Scheme} or the \textsf{Linear Combination Scheme}, we can compute a value $\omega$ such that extraction is possible whenever $\succp>\frac{\omega-1}{\gamma}$.  We can calculate $\succp$ for a given proving algorithm $\cal P$
if we know the values of $\cal P$'s response ${\cal P}(c)$ for
every possible challenge $c\in\Gamma$.  However, the whole
purpose of a \POR \ scheme is to provide reassurance that $\succp$ is
sufficiently large without having to request all ${\cal P}(c)$ for
all $c\in\Gamma$.  Given the prover's responses to some
subset of possible challenges, the user wishes to make a judgement
as to whether he/she is satisfied that $\succp$ is acceptably
high.  This takes us straight into the realm of classical statistical techniques
such as {\em hypothesis testing} \cite{CB,dalyetal}.

Suppose the prover has given responses ${\cal P}(c)$ to $t$
challenges $c$ chosen uniformly at random without replacement from
$\Gamma$, and that $g$ of these responses are found to be
correct.  We are concerned that the prover's success rate may not be high
enough, so we are looking for evidence to convince us that in fact it is
sufficiently high to permit extraction.  In other words, we wish to distinguish the null hypothesis
\begin{description}
\item[$H_0:$] $\succp\leq\frac{\omega-1}{\gamma};$
\end{description}
from the alternative hypothesis
\begin{description}
\item[$H_1:$] $\succp\geq\frac{\omega}{\gamma}.$
\end{description}
Suppose that $H_0$ is true.  Then the probability that the number
of correct responses is at least $g$ is itself  at most
\begin{equation}
\sum_{i=g}^t\frac{\binom{\omega-1}{i}\binom{\gamma-\omega+1}{t-i}}{\binom{\gamma}{t}}.\label{eq:prob}
\end{equation}
If this probability is less than $0.05$, then we reject $H_0$ and
instead accept the alternative hypothesis (namely that $\succp$ is
sufficiently high to permit extraction. In this case we conclude that the server is
 storing the file appropriately.) If the probability is greater
than $0.05$ then there is insufficient evidence to reject $H_0$ at
the $5\%$ significance level (so we continue to suspect that the
server is perhaps not storing the file adequately).

Alternatively, if we choose the challenges uniformly at random
{\em with} replacement, then the condition for rejecting the null
hypothesis becomes
\begin{equation*}
\sum_{i=g}^t\binom{t}{i}\left(\frac{\omega-1}{\gamma}\right)^i\left(\frac{\gamma-\omega+1}{\gamma}\right)^{t-i}<0.05.
\end{equation*}

\begin{example}\label{ex:significance}
Suppose that for the \textsf{Basic Scheme} $n=1000$, and that the minimum
distance of the response code is $400$.  Then by Theorem~\ref{t1.1} we find that
extraction is possible whenever $\succp$ is greater than $0.8$.  Suppose the
prover responds to $100$ challenges that have been chosen uniformly with
replacement, and that $87$ of the responses were correct.
We find that
\begin{align*}
\sum_{i=87}^{100}\binom{100}{i}0.8^i0.2^{100-i}\approx 0.047<0.05.
\end{align*}
Thus, in this case there is sufficient evidence to reject the null hypothesis at
the $5\%$ significance level, and so we conclude that the file is in
fact being stored correctly.

On the other hand, if only $86$ of the responses were correct, we observe that
\begin{align*}
\sum_{i=86}^{100}\binom{100}{i}0.8^i0.2^{100-i}\approx 0.08 >0.05.
\end{align*}
In this case there is not enough evidence to reject the null hypothesis at the
$5\%$ significance level, and so we continue to suspect that the server is not storing the file adequately.
\end{example}
The benefit of this statistical approach is that
given the observed responses to the challenges, for any desired value of
$\alpha$ we can construct a hypothesis
test for which the probability of inappropriately rejecting the null hypothesis
(and hence failing to catch a prover that does not permit extraction) is
necessarily less than $\alpha$.\footnote{Here we refer to probability over the set of all
possible choices of $t$ challenges.}  This is the case regardless of the true
value of $\succp$, and we do not need to make any { a priori} assumptions
about this value.

In Table~\ref{tab:rejection} we give examples of a range of possible results of
the challenge process and the corresponding outcomes in terms of whether the
null hypothesis is rejected at either the $5\%$ or $1\%$ significance level.

\subsection{Confidence Intervals}\label{subsec:conf}
Another closely-related way to portray the information provided by the sample of responses to challenges is through the use of {confidence intervals.}
We define a $95\%$ lower confidence bound $\theta_L$ by
\begin{equation*}
\theta_L=\sup\left\{\theta\bigg\vert \sum_{i=g}^t
\binom{t}{i}\theta^{i}(1-\theta)^{t-i} <0.05\right\}
\end{equation*}
This represents the largest possible value for $\succp$ for which the
probability of obtaining $g$ or more correct responses in a sample of size $t$
is less than $0.05$.  Then the decision process for the hypothesis test
described in Section~\ref{subsec:hypothesis} consists of rejecting the null hypothesis whenever
$\frac{\omega-1}{n}<\theta_L$, since if $\frac{\omega-1}{n}$ is less than the
critical value we know that the probability of a prover with success rate at
most $\frac{\omega-1}{n}$ providing $g$ or more correct responses is less than
$0.05.$

  The interval $(\theta_L,1]$ is a {\em $95\%$ confidence interval for $\succp$}:  if a large number of samples of size
$t$ were made and the corresponding intervals were calculated using this
approach, then you would expect the resulting intervals to contain the true
value of $\succp$ at least $95\%$ of the time.  The hypothesis test can be
expressed in terms of the confidence interval $(\theta_L,1]$ by stating that we
reject $H_0$ whenever $\frac{\omega-1}{n}$ does not lie in this interval.

\begin{example}
Suppose we have $n=1000$ and $d=400$ as in Example~\ref{ex:significance}, and
suppose that $90$ of the responses are correct.  Then
\begin{align*}
\theta_L&=\sup\left\{\theta\bigg\vert \sum_{i=90}^t
\binom{t}{i}\theta^{i}(1-\theta)^{t-i} <0.05\right\}\\&\approx0.836
\end{align*}
Then a $95\%$ confidence interval for $\succp$ is $(0.836,1]$, and hence we
reject the null hypothesis, as $\frac{\omega-1}{n}=0.8$ does not lie in this
interval.
\end{example}

\begin{table}[htb]
{\small\begin{align*}
\begin{array}{|ccccc||ccccc|}
\hline \frac{\omega-1}{\gamma}&t&g&\alpha=0.05&\alpha=0.01&\frac{\omega-1}{\gamma}&t&g&\alpha=0.05&\alpha=0.01\\
\hline
0.8&100&100&\checkmark&\checkmark&0.9&100&100&\checkmark&\checkmark
\\
0.8&100&95&\checkmark&\checkmark&0.9&100&95&\sf{X}&\sf{X}
\\
0.8&100&90&\sf{X}&\sf{X}&0.9&100&90&\sf{X}&\sf{X}
\\
0.8&100&85&\sf{X}&\sf{X}&0.9&100&85&\sf{X}&\sf{X}
\\
0.8&100&80&\sf{X}&\sf{X}&0.9&100&80&\sf{X}&\sf{X}
\\
\hline
0.8&200&180&\checkmark&\checkmark&0.9&200&200&\checkmark&\checkmark
\\
0.8&200&175&\checkmark&\checkmark&0.9&200&195&\checkmark&\checkmark
\\
0.8&200&170&\checkmark&\sf{X}&0.9&200&190&\checkmark&\checkmark
\\
0.8&200&165&\sf{X}&\sf{X}&0.9&200&185&\sf{X}&\sf{X}
\\
0.8&200&160&\sf{X}&\sf{X}&
0.9&200&180&\sf{X}&\sf{X}
\\
\hline
0.8&500&435&\checkmark&\checkmark&0.9&500&480&\checkmark&\checkmark
\\
0.8&500&430&\checkmark&\checkmark&0.9&500&475&\checkmark&\checkmark
\\
0.8&500&425&\checkmark&\checkmark&0.9&500&470&\checkmark&\checkmark
\\
0.8&500&420&\checkmark&\sf{X}&0.9&500&465&\checkmark&\sf{X}
\\
0.8&500&415&\sf{X}&\sf{X}&0.9&500&460&\sf{X}&\sf{X}
\\
\hline
\end{array}
\end{align*}}
\caption{Outcomes of hypothesis testing for a range of responses.  The
columns headed by values of $\alpha$ contain a tick if $H_0$ is rejected at the
corresponding significance level, and a cross otherwise.}
\label{tab:rejection}
\end{table}

\subsection{Reacting to a Suspect Prover}
One question that has not always been directly considered is what action to take
when a prover is suspected of cheating.  In the framework of
Section~\ref{subsec:hypothesis} this becomes the problem of what to do in the
case where there is insufficient evidence to reject the null hypothesis.  There
are various possible options at this point, and the choice of option will depend
on factors such as the reason for storing the file, and any costs and
inconvenience that might be associated with associated with ceasing to use that
server, or with switching to another storage provider.  For example, if a server
is simply being used as a backup service for non-critical data and there
is a high overhead associated with switching storage providers, then a user will
not want to be overhasty in taking action against a possibly innocent server.
In this case an appropriate action in the first instance might be to seek more responses to
challenges in order to avoid the possibility that the earlier set of responses
were unrepresentative of the reliability of the prover in general.

\subsection{Comparison with Approaches Followed in the Literature}
Ateniese et al.\ \cite{Aten} observe that if
$\succp\leq\frac{g}{\gamma}$ (with $g\in\mathbb{Z}$) then when ${\cal P}$ is queried on $t$ possible challenges chosen
uniformly at random without replacement then the probability $\sf Pr_{bad}$ that at least one incorrect response is
observed is given by
\begin{align*}{\sf
Pr_{bad}}&=1-\frac{\binom{g}{t}}{\binom{\gamma}{t}},\intertext{and they note that}
1-\left(\frac{g}{\gamma}\right)^t&\leq {\sf
Pr_{bad}}\leq
1-\left(\frac{g+1-t}{\gamma-t+1}\right)^t.
\end{align*}
As examples of parameters, they
point out that if $\succp=0.99$ then to achieve ${\sf
Pr_{bad}}=0.95$ requires $t=300$, and ${\sf Pr_{bad}}=0.99$
requires $t=460$.  They comment that the required number $t$ is in fact
independent of $\gamma$, sinced it is based instead on the required threshold
for $\succp$ -this observation applies equally to our analysis.

Dodis, Vadhan and Wichs \cite{DVW} use a similar approach to
Ateniese {et al.}, but in addition propose the use of a {\em
hitting sampler} that amounts to choosing which $t$ elements to
sample from a specified distribution that contains fewer than
$\binom{n}{t}$ possible sample sets but still guarantees that $\sf
Pr_{bad}$ is higher than some specified value for a given value of $\succp$ that
is less than 1.

This analysis can be interpreted in the context of a hypothesis test to
distinguish the null hypothesis
\begin{description}
\item[$H_0:$] $\succp\leq0.99;$
\end{description}
from the alternative hypothesis
\begin{description}
\item[$H_1:$] $\succp>0.99$.
\end{description}
If 300 challenges are made and all the responses are correct, then a $95\%$
confidence interval for $\succp$ is $(0.99006,1]$, so there is enough evidence
to reject the null hypothesis at the $5\%$ significance level.  However, a
$99\%$ confidence interval for $\succp$ is $(0.977,1]$, so there is insufficient
evidence to reject the null hypothesis at the $1\%$ significance level.  If, on
the other hand, 460 challenges were made and all the responses were correct then
a $99\%$ confidence interval for $\succp$ is $(0.99003,1]$ and so in this case
there {\em is} enough evidence to reject the null hypothesis at the $1\%$
significance level.  

We note that this is a special case of the analysis in
Sections~\ref{subsec:hypothesis} and \ref{subsec:conf}.  Specifically, Ateniese
et al.\ are focusing on determining the smallest number of challenges for
which an entirely correct response constitutes sufficient evidence to reject the
null hypothesis at the desired significance level.  While this does result in
the smallest number of challenges for which there is still the potential to reassure the user as to the appropriate behaviour
of the prover, it has the drawback that even a single
incorrect response results in failure to reject the null hypothesis, regardless
of whether it is true.  Taking a larger sample size has the benefit of
increasing the probability that a false
null hypothesis is rejected, without adversely affecting the probability that a true
$H_0$ fails to be rejected. For example, from Table~\ref{tab:rejection} we see
that if 90 correct responses out of 100 are observed then there is insufficient
evidence to reject the hypothesis $\succp\leq 0.8$ at the $5\%$ significance level.  However, if 180 correct responses out of 200 are observed
then there {\em is} sufficient evidence to reject this null hypothesis at the
$5\%$ significance level (in fact we even have enough evidence to reject $H_0$
at the $1\%$ significance level). 

\section{A Lower Bound on Storage and Communication Requirements}
\label{lower.sec}

In this section, we prove a bound that applies to {\it keyed} \POR \ schemes.
{Suppose that $\mathbf{M}$ is a random variable corresponding to a 
randomly chosen {\it unencoded} message $m$. Let $\mathbf{V}$ be a random variable
denoting the information stored by the \Ver \ (i.e., the key), and let $\mathbf{R}$ be a random variable
corresponding to 
the computations performed by an extractor. It is obvious that the probability that $m$ can correctly
be reconstructed is $2^{-H( \mathbf{M} | \mathbf{V}, \mathbf{R} )}$.
Now, from basic entropy inequalities, we have
\begin{eqnarray*}
H( \mathbf{M} | \mathbf{V}, \mathbf{R} ) & = & 
H( \mathbf{M},\mathbf{V}, \mathbf{R} ) - H( \mathbf{V}, \mathbf{R} ) \\
& \geq & H( \mathbf{M},\mathbf{V}, \mathbf{R} ) - H( \mathbf{V}) - H( \mathbf{R} )\\
& \geq & H( \mathbf{M}) - H( \mathbf{V}) - H( \mathbf{R} ).
\end{eqnarray*}
} 
{Suppose that the message can be reconstructed by the extractor with probability $1$. Then we have
$H( \mathbf{M} | \mathbf{V}, \mathbf{R} ) = 0$. The inequality proven above imples that
\begin{equation} 
\label{h1.eq}
H( \mathbf{M}) \leq H( \mathbf{V}) + H( \mathbf{R} ).
\end{equation}
Now suppose that the extractor is a black-box extractor.
In this situation, we have that
\begin{equation} 
\label{h2.eq} H( \mathbf{R} ) \leq \gamma \log_2 |\Delta|,
\end{equation}
since there are $\gamma$ possible challenges and each response is from 
the set $\Delta$.  
The message $m$ is a random vector in $(\eff_q)^k$, so 
\begin{equation} 
\label{h3.eq} H( \mathbf{M}) = k \log _2 q.
\end{equation}
Therefore, combining (\ref{h1.eq}), (\ref{h2.eq}) and (\ref{h3.eq}), we have the following result.}

{
\begin{theorem} Suppose we have a keyed \POR \ scheme where the message is a random 
vector in $(\eff_q)^k$, there are $\gamma$ possible challenges and each response is from 
the set $\Delta$. Suppose that a black-box extractor succeeds with probability 
equal to $1$. The the entropy of the verifier's storage, 
denoted  $H( \mathbf{V}) $, satisfies the
inequality
\[ H( \mathbf{V}) \geq k \log _2 q - \gamma \log_2 |\Delta|.\]
\end{theorem}
}

{We mentioned previously that Naor and Rothblum \cite{NR} proved
a lower bound for a weaker form of \POR-type protocol, termed an 
``authenticator''. As noted in \cite{DVW}, the Naor-Rothblum 
bound also applies to \POR \ schemes.
Phrased in terms of entropy, their bound
states that 
\[ H( \mathbf{M}) \leq H( \mathbf{V}) \times H( \mathbf{R} ),\]
which is a weaker bound than (\ref{h1.eq}). 
}

In the case of an unkeyed scheme, the extractor is only given access to
the proving algorithm. Therefore, 
$H( \mathbf{M} | \mathbf{R} ) = 0$ if a black-box extractor succeeds with probability 
equal to $1$. From this, it follows that 
$H( \mathbf{M}) \geq H(\mathbf{R} )$ in this situation.

\section{Conclusion} We have performed a comprehensive analysis of the
extraction properties of unconditionally secure \POR \ schemes, and established
a methodology that is applicable to the analysis of further new schemes.  What
constitutes ``good'' parameters for such a scheme depends on the precise
application, but our framework allows a flexible trade-off between parameters.
One direction of possible {future} interest would be to consider the construction
of further keyed \POR\ schemes with a view to reducing the user's key storage
requirements.

\section*{Acknowledgements} Thanks to Andris Abakuks and Simon Skene for some
helpful discussions of statistics.

\appendix


\section*{Appendix}

\begin{table}[htb]
\caption{Notation used in this paper}
\label{notation.tab}
\begin{center}
\begin{tabular}{l|l}
$q$ & order of underlying finite field\\ \hline
$m$ & message \\ \hline
$m_i$ & message block \\ \hline
$\mathcal{M}$ & message space \\ \hline
$k$ & length of a message \\ \hline
$M$ & encoded message \\ \hline
$\mathcal{M}^*$ & encoded message space \\ \hline
$n$ & length of an encoded message \\ \hline
$d$ & distance of the encoded message space \\ \hline
$c$ & challenge \\ \hline
$\Gamma$ & challenge space \\ \hline
$\gamma$ & number of possible challenges \\ \hline
$r$ & response \\ \hline
$\rho$  & response function \\ \hline
$r^M$ & response vector for encoded message $M$\\ \hline
$\Delta$ & response space \\ \hline
$\mathcal{R}^*$ & response code \\ \hline
$d^*$ & distance of the response code \\ \hline
$\mathcal{P}$ & proving algorithm \\ \hline
$\mathsf{succ}(\mathcal{P})$ & success probability of proving algorithm \\ \hline
$\widehat{m}$ & message outputted by the \Ext \\ \hline
$K$ & key (in a keyed scheme) \\ \hline
$S$ & tag (in a keyed scheme) \\ \hline
$\dist $ & hamming distance between two vectors
\end{tabular}
\end{center}
\end{table}


\begin{thebibliography}{X}

\bibitem{Aten}
G. Ateniese, R.C. Burns, R. Curtmola, J. Herring, L. Kissner, Z.N.J. Peterson and D.X. Song.
Provable data possession at untrusted stores.
{\it ACM Conference on Computer and Communications Security} (2007) 598--609.

\bibitem{BEGKN94}
{M. Blum, W.S. Evans, P. Gemmell, S. Kannan, and M. Naor.
Checking the correctness of memories.
{\em Algorithmica} {\bf 12} (1994), 225--244.}


\bibitem{BJO}
K.D. Bowers, A. Juels and A. Oprea.
Proofs of retrievability: theory and implementation.
{\it Proceedings of the first ACM Cloud Computing Security Workshop}
(2009), 43--54.

\bibitem{CB} G. Casella and R.L. Berger. {\it Statistical Inference}, Duxbury Press, 1990.

\bibitem{dalyetal} F. Daly, D.J. Hand, M.C. Jones, A.D.Lunn, K.J. McConway.
{\it Elements of Statistics}, Addison-Wesley, 1995.

\bibitem{DVW} Y. Dodis, S. Vadhan and D. Wichs.
Proofs of retrievability via hardness amplification.
{\it Lecture Notes in Computer Science} {\bf 5444} (2009), 109--127
(TCC 2009).

\bibitem{HKRU}
M. I. Husain, S. Ko, A. Rudra and S. Uurtamo.
Storage enforcement with Kolmogorov complexity and list decoding.
ArXiv report 1104/3025, 2011,
{\tt http://arxiv.org/abs/1104.3025}.


\bibitem{JK}
A. Juels and B.S. Kaliski Jr.
PORs: proofs of retrievability for large files.
{\it ACM Conference on Computer and Communications Security} (2007), 584--597.

\bibitem{LEBB}
{M. Lillibridge, S. Elnikety, A. Birrell and M. Burrows.
A cooperative internet backup scheme.
{\it USENIX} (2003), 29--41.}


\bibitem{MS}
F.J. MacWilliams and N.J.A. Sloane.
{\it The Theory of Error-Correcting Codes}, North-Holland, 1977.

\bibitem{NR} {M. Naor and G.N. Rothblum.
The complexity of online memory checking.
{\it Journal of the ACM} {\bf 56} (2009).}


\bibitem{SW} H. Shacham and B. Waters. Compact proofs of retrievability.
{\it Lecture Notes in Computer Science} {\bf 5350} (2008),  90--107
(ASIACRYPT 2008).

\end{thebibliography}
\end{document}